\documentclass[twocolumn,a4paper,accepted=2020-09-01]{quantumarticle}
\pdfoutput=1

\usepackage[numbers,sort&compress]{natbib}
\usepackage{graphicx}
\usepackage[utf8]{inputenc}
\usepackage[english]{babel}
\usepackage[T1]{fontenc}
\usepackage{amsmath}
\usepackage{hyperref}
\hypersetup{breaklinks=true}
\usepackage{mathtools}
\usepackage{mathrsfs}
\usepackage{physics}
\usepackage{amsthm}
\usepackage{amssymb}
\usepackage{dsfont}
\usepackage{enumitem}
\newcommand{\diff}[1]{\ensuremath{\operatorname{d}\!{#1}}}
\usepackage{color}
\newtheorem{thm}{Theorem}
\newtheorem{lem}[thm]{Lemma}

\begin{document}


\title{Worst-case Quantum Hypothesis  Testing with Separable Measurements}

\author{Le Phuc Thinh}
\email{thinh.le@itp.uni-hannover.de}
\affiliation{Centre for Quantum Technologies, National University of Singapore, Singapore}
\affiliation{Institut f{\"u}r Theoretische Physik, Leibniz Universit{\"a}t Hannover, Appelstr. 2, 30167 Hannover, Germany}

\author{Michele Dall'Arno}
\email{dallarno.michele@yukawa.kyoto-u.ac.jp}
\affiliation{Centre for Quantum Technologies, National University of Singapore, Singapore}
\affiliation{Yukawa Institute for Theoretical Physics, Kyoto University, Kitashirakawa Oiwakecho, Sakyoku, Kyoto 606-8502, Japan}
\affiliation{Faculty of Education and Integrated Arts and Sciences, Waseda University, 1-6-1 Nishiwaseda, Shinjuku-ku, Tokyo 169-8050, Japan}

\author{Valerio Scarani}
\email{physv@nus.edu.sg}
\affiliation{Centre for Quantum Technologies, National University of Singapore, Singapore}
\affiliation{Department of Physics, National University of Singapore, Singapore}

\begin{abstract}
For any pair of quantum  states (the hypotheses), the task of  binary  quantum hypotheses  testing  is  to derive  the tradeoff   relation  between   the  probability $p_{01}$ of rejecting  the  null  hypothesis  and $p_{10}$ of accepting  the alternative  hypothesis.  The case when both hypotheses  are  explicitly  given was solved in the pioneering work by Helstrom. Here, instead, for any  given null hypothesis as a pure state, we consider the worst-case  alternative hypothesis that  maximizes $p_{10}$  under a  constraint  on  the distinguishability of such  hypotheses.  Additionally, we restrict the optimization to separable measurements, in order to describe tests that are performed locally. The case $p_{01}=0$ has been recently studied under the name of ``quantum state verification''. We show that the problem can be cast as a semi-definite program (SDP). Then we study in detail the two-qubit case. A comprehensive study in parameter space is done by solving the SDP numerically. We also obtain analytical solutions in the case of commuting hypotheses, and in the case where the two hypotheses can be orthogonal (in the latter case, we prove that the restriction to separable measurements generically prevents perfect distinguishability). In regards to quantum state verification, our work shows the existence of more efficient strategies for noisy measurement scenarios.
\end{abstract}


\rightline{YITP-20-101}

\maketitle

\section{\label{sec:intro}Introduction}

The         task        of         quantum        hypothesis
testing~\cite{helstrom1969quantum}, a subfield of quantum state estimation~\cite{Hradil97, Paris_book, Holevo_book}, is to optimally identify,
according to some given  payoff function, an unknown quantum
state  given  as a  black  box.   The strategy  consists  of
performing a  quantum measurement, whose  optimality depends
upon the payoff function and the prior information about the
state, available  in the form of  a probability distribution
over    the    state    space.     Several    discrimination
problems~\cite{discriminate_unitary,   discriminate_channel,
  quantum_reading, BDD11, DBDMJD12, qillumination, quantumstatediscrimination, discriminate_mixedstate, discriminate_errormargin, maxconfidence} are based
on quantum hypothesis testing. 

In  the simplest  non-trivial instance  of the  problem, the
prior  distribution has  support over  two states  only, the
\textit{null}   and  the   \textit{alternative}  hypotheses.
Hence, binary quantum hypothesis  testing corresponds to the
derivation  of  the  tradeoff  relation  between  two  error
probabilities:  I) the  probability  of  rejecting the  null
hypothesis,  and  II)  the   probability  of  accepting  the
alternative hypothesis.   The case when such  hypotheses are
given    explicitly     was    solved     analytically    by
Helstrom~\cite{helstrom1969quantum}.

Here, instead, we consider the  case in which only one state
(the  null hypothesis)  is explicitly  given. For  the other
state   (the   alternative   hypothesis),  we   consider   a
constrained worst-case scenario.  The worst-case alternative
hypothesis  is  the one  that  maximizes  the type-II  error
probability,   under   a   given    lower   bound   on   the
distinguishability of the  two hypotheses.  Additionally, we
consider  multipartite  hypotheses,   and  we  restrict  the
optimization  to separable  measurements  only.  This  setup
generalizes      the      so-called     ``quantum      state
verification''~\cite{qv,     qv_bipartite,    qv_dicke,        qv_2qubit,
  qv_correspondence,  qv_noniid, qv_pure},  by relaxing  the
assumption that the type-I error probability is null.

Our   first   contribution   is   a   formulation   of   the
aforementioned problem as a  semi-definite program, that can
be  efficiently  solved  with  readily  available  numerical
tools. Then, we specify to  the case in which the hypotheses
are two-qubit states,  and we derive an  analytical form for
the worst-case hypothesis.   Finally, we analytically derive
the  tradeoff  relation  between type-I  and  type-II  error
probabilities in the case when the hypotheses commute. In regards to quantum state verification, our work shows the existence of more efficient strategies in certain parameter regimes for noisy measurement scenarios.

The   structure   of  this   paper   is   as  follows.    In
Section~\ref{sec:hypothesis_testing}  we recall  the general
problem of  quantum hypothesis testing and  we introduce the
specific problem  addressed here.   In Section~\ref{sec:sdp}
we reformulate  our problem as a  semi-definite program, and
we analytically derive the worst-case alternative hypothesis
in  the  two-qubit  case.  In  Section~\ref{sec:results}  we
analytically derive the tradeoff relation between type-I and
type-II  error probabilities  for commuting  hypotheses. Section~\ref{sec:conclusion} summarizes our results.

\section{\label{sec:hypothesis_testing}Hypothesis testing of quantum states}

The simplest scenario of hypothesis testing is binary {\em quantum state discrimination} between $\rho_0$ and $\rho_1$. In other words, one is asked to decide which is more likely between two hypotheses $H_0$ --- the {\em null hypothesis} --- representing the fact that the unknown state is $\rho_0$, and $H_1$ --- the {\em alternative hypothesis} --- corresponding to the unknown state being $\rho_1$. The decision process can be formalized by a POVM $\{\Omega,\openone-\Omega\}$ where the element $\Omega$ accepts $H_0$ and $\openone-\Omega$ accepts $H_1$. This naturally gives rise to two errors, type I or {\em false positive} $p_{01}=\tr(\rho_0(\openone-\Omega))$ and type II or {\em false negative} $p_{10}=\tr(\rho_1\Omega)$. False positive probability captures the situation that the decision process accepts $H_1$ when hypothesis $H_0$ is true. False negative probability corresponds to the other situation where one accepts $H_0$ when $H_1$ is true. Therefore, in this language, the problem is to design an {\em optimal measurement} $\Omega$ that optimizes certain figure-of-merit. For example, Helstrom strategy minimizes the average probability of error $p_0p_{01}+p_1p_{10}$ where $p_0$ is the {\em a priori} probability of occurrence of hypothesis $H_0$ and ditto for $p_1$.

Several problems in quantum information such as quantum channel coding~\cite{channelcoding} and quantum illumination~\cite{qillumination} can be seen as hypothesis testing problems by assigning appropriate sets to hypothesis and choosing appropriate figures-of-merit (see e.g.~\cite{hayashi_2006,hayashi_2009}). Here we look at the task that has been called {\em quantum state verification}~\cite{qv}. In this task, $H_0$ is the state $\ketbra{\psi}$, and $H_1$ is the set of states $\sigma$ such that $\bra{\psi}\sigma\ket{\psi}\leq1-\epsilon$. Previous works \cite{qv, qv_bipartite, qv_dicke, qv_2qubit, qv_correspondence,  qv_noniid, qv_pure} considered strategies that have no false positive, i.e. $p_{01}=0$, and set out to minimize the worst-case probability of false negative
$$
p_{10}(\epsilon):=\min_{\substack{0\preceq\Omega\preceq\openone \\ \bra{\psi}\Omega\ket{\psi}=1\\ \Omega\in\text{[set]}}}\max_{\substack{\sigma\succeq0 \\ \tr(\sigma)=1 \\ \bra{\psi}\sigma\ket{\psi}\leq1-\epsilon}}\tr(\Omega\sigma)\,.
$$ The set to which $\Omega$ belongs can be that of all effects, or a restricted one. When dealing with composite systems, a particularly relevant set is the set $\text{SEP}$ of \textit{separable} measurements because they are easier to implement than LOCC or richer local measurement classes and at the same time could provide a bound on the performance of other classes. In this work, we shall focus on this one and leave possible extensions to future work.

Here we relax the condition $p_{01}=0$ to $p_{01}\leq\delta$, leading to the optimisation
\begin{align}\label{eq:main_opt_sep}
p_{10}(\delta,\epsilon):=\min_{\substack{0\preceq\Omega\preceq\openone \\ \bra{\psi}\Omega\ket{\psi}\geq1-\delta \\ \Omega\in\text{SEP}}}\max_{\substack{\sigma\succeq0 \\ \tr(\sigma)=1 \\ \bra{\psi}\sigma\ket{\psi}\leq1-\epsilon}}\tr(\Omega\sigma)\,.
\end{align}
This generalisation is relevant, as it allows the study of the \textit{tradeoff} between $\delta$ and $p_{10}(\delta,\epsilon)$. From the technical point of view, this study does not constitute a straightforward extension of previously employed mathematical tools for the following reason. The condition $\bra{\psi}\Omega\ket{\psi}=1$ forces $\Omega$ to commute with $\ketbra{\psi}$, which provides a significant simplification in the number of parameters and structure of the problem. When that condition is relaxed to $\bra{\psi}\Omega\ket{\psi}\geq 1-\delta$, commutativity can no longer be assumed \textit{a priori} (and we shall show that, for some values of $\delta$ and the other parameters, the optimal strategy is indeed \textit{not} the commuting one).

\section{\label{sec:sdp}Reformulations of the optimisation}

In this section, we first show that the optimisation \eqref{eq:main_opt_sep} for separable measurements can be cast as a semidefinite program (SDP), which allows for reliable numerical solutions. Then, for the case of two-qubit states, we solve the optimisation of the inner problem, thus casting the optimisation in a form which will allow deriving some analytical results in Section \ref{sec:results}.

\subsection{Reformulation as a SDP}
The problem we are considering is at first sight a min max problem involving two variables $\Omega,\sigma$ that appears bilinearly in the objective function. Though fixing each variable is separately a SDP and can be reliably solved to any precision, there is no guarantee on the optimality of remaining outer optimization if one deploys numerical methods. A closer analysis of the optimization problem shows that one can in fact use duality theory of semidefinite programming to reformulate the problem. We refer the reader to the classic book~\cite{boyd2004convex} for more information on duality in optimization.

\begin{lem}
The optimisation \eqref{eq:main_opt_sep} can be reformulated as a semidefinite program
\begin{equation}\label{eq:main_opt_sep_sdp}
\begin{aligned}
p_{10}(\delta,\epsilon)=&\min_{\Omega,y_1,y_2} y_1+(1-\epsilon)y_2 \\
&\text{ s. t. } 0\preceq\Omega\preceq\openone\\
&\phantom{\text{ s. t. }} \bra{\psi}\Omega\ket{\psi}\geq1-\delta\\
&\phantom{\text{ s. t. }} \Omega\in\textup{SEP}\\
&\phantom{\text{ s. t. }} y_1\openone+y_2\ketbra{\psi}\succeq\Omega \\
&\phantom{\text{ s. t. }}  y_1\in\mathbb{R}, y_2\geq0 \\
\end{aligned}
\end{equation}
\end{lem}
\begin{proof}
The constraints on $\Omega$ (outer optimisation) remain the same, while we replace the inner optimisation
$$\max\{\tr(\Omega\sigma):\sigma\succeq0,\tr\sigma=1,\bra{\psi}\sigma\ket{\psi}\leq1-\epsilon\}$$
by its dual, which is the semidefinite program
$$\min\{y_1+(1-\epsilon)y_2:y_1\openone+y_2\ketbra{\psi}\succeq\Omega^\dagger,y_2\geq0\}\,.$$ Moreover, strong duality holds because the primal is feasible, and thanks to $\Omega^\dagger=\Omega$ the dual is strictly feasible (choose $y_2>0$ such that $(y_1\openone-\Omega+y_2\ketbra{\psi}\succ0$). This means that the primal and dual optimum are the same, and also the primal optimum is attained. Hence, our minimax problem becomes~\eqref{eq:main_opt_sep_sdp}. Note that separability is a SDP constraint albeit exponential in size~\cite{separable_SDP} and not just a hierarchy of SDP constraints~\cite{DPS_hierarchy}.
\end{proof}

\subsection{Two-qubit states}

For two-qubit states $\ket{\psi}=\cos\theta\ket{00}+\sin\theta\ket{11}$, we can proceed with additional analytic derivations. Without loss of generality, we consider the regime of parameters where $\theta\in[0,\pi/4]$, $\epsilon\in(0,1]$ and $\delta\in[0,1]$.

\begin{lem}\label{lemma2}
For two-qubit pure states $\ket{\psi}=\cos\theta\ket{00}+\sin\theta\ket{11}$, the optimisation~\eqref{eq:main_opt_sep_sdp} reduces to an optimisation over real variables
\begin{equation}\label{eq:optlemma2}
\begin{aligned}
p_{10}(\delta,\epsilon)=&\min_{t,z,x,\omega,y_1,y_2} y_1+(1-\epsilon)y_2 \\
&\text{ s. t. } 0\preceq\begin{pmatrix}
t+z & x\\
x & t-z
\end{pmatrix}
\preceq\openone\\
&\phantom{\text{ s. t. }} 0\leq\omega\leq1\\
&\phantom{\text{ s. t. }} t+z\geq1-\delta\\
&\phantom{\text{ s. t. }} \omega\geq\abs{x\cos2\theta+z\sin2\theta}\\
&\phantom{\text{ s. t. }} \begin{pmatrix}
y_1+y_2-(t+z) & -x\\
-x & y_1-(t-z)
\end{pmatrix}\succeq0 \\
&\phantom{\text{ s. t. }} y_1-\omega\geq0 \\
&\phantom{\text{ s. t. }}  y_1\in\mathbb{R}, y_2\geq0 \\
\end{aligned}
\end{equation}
\end{lem}
\begin{proof}
We first spend the symmetry present in the state. Define \begin{align}
\Omega_a:=\frac{1}{2\pi}\int_0^{2\pi}(U_{\phi} \otimes U_{-\phi})\Omega(U_{\phi} \otimes U_{-\phi})^\dagger \diff\phi
\end{align}
with $U_{\phi}=\ketbra{0}+e^{i\phi}\ketbra{1}$. For any feasible $(\Omega,y)$, the pair $(\Omega_a,y)$ remains feasible with the same value of the objective function. Moreover, the state is also invariant under swapping $S$ of two qubits, so that $(\bar{\Omega}_a,y)$ with $\bar{\Omega}_a:=(\Omega_a+S\Omega_aS^\dagger))/2$ is feasible as well. Lastly, $\bar{\Omega}_a$ can be taken to be real symmetric because the feasible region is preserved under taking entrywise complex conjugate, and the objective value is unchanged. We note that this same symmetrisation was carried out in \cite{qv_2qubit} on the primal inner problem, thanks to the assumption that $\Omega$ commutes with $\ket{\psi}\bra{\psi}$. It's by looking at the dual that we noticed that the symmetry is independent of the commutation assumption.

This observation simplifies the number of variables in our optimisation. Specifically, let $\ket{\psi^\perp}=-\cos\theta\ket{00}+\sin\theta\ket{11}$, it suffices to optimize over real symmetric matrices
\begin{align}
\bar{\Omega}_a = \begin{pmatrix}
t+z & x & 0 & 0\\
x & t-z & 0 & 0\\
0 & 0 & \omega & 0\\
0 & 0 & 0 & \omega
\end{pmatrix}
\end{align}
in the ordered basis $\{\ket{\psi},\ket{\psi^\perp},\ket{01},\ket{10}\}$. Writing out the separability constraint, which for qubits is equivalent to positive partial transpose, we arrive at the final form given by the Lemma.
\end{proof}

Remarkably, what was the inner optimisation (now optimisation over $y_1$ and $y_2$) can be further solved analytically; besides, one can set $t=1-\delta-z$ and $\omega=\abs{x\cos2\theta+z\sin2\theta}$ without loss of generality. The lengthy proof of these steps is presented in Appendix \ref{sslengthy}. The farthest version of the optimisation that we can reach analytically reads: 
\begin{lem}\label{lemmafinal}
For two-qubit pure states $\ket{\psi}=\cos\theta\ket{00}+\sin\theta\ket{11}$, the optimisation~\eqref{eq:main_opt_sep_sdp} reduces to
\begin{equation}\label{eq:final}
\begin{aligned}
p_{10}(\delta,\epsilon)=&\min_{z,x} f(\abs{x\cos2\theta+z\sin2\theta}) \\
&\text{ s. t. } 0\preceq\begin{pmatrix}
1-\delta & x\\
x & 1-\delta-2z
\end{pmatrix}
\preceq\openone\\
&\phantom{\text{ s. t. }} 1-\delta-2z+\sqrt{\frac{1-\epsilon}{\epsilon}}\abs{x}\leq\abs{x\cos2\theta+z\sin2\theta}
\end{aligned}
\end{equation}
where
\begin{equation*}
f(y_1^*):=y_1^*+(1-\epsilon)\left[1-\delta-y_1^*+\frac{x^2}{y_1^*-(1-\delta-2z)}\right]\,.
\end{equation*}
\end{lem}

\section{Results}
\label{sec:results}

For our results, we keep focusing on the case of two qubits, although we recall that the SDP \eqref{eq:main_opt_sep_sdp} is valid in general and one could therefore set out to solve it in any other case.

\subsection{Commuting strategy}

As we mentioned earlier, when $\delta=0$, which is the case considered in~\cite{qv,qv_2qubit}, the condition $\bra{\psi}\Omega\ket{\psi}=1$ immediately implies that $\Omega$ is diagonal in the same basis as $\ket{\psi}\bra{\psi}$, that is $x=0$ in our notation. When $\delta\neq 0$, there is no \textit{a priori} guarantee that the optimal solution will be a commuting one; but we can obtain an upper bound $p^c_{10}(\delta,\epsilon)$ by \textit{enforcing} $x=0$. In this case, the optimisation \eqref{eq:final} becomes trivial:
\begin{equation*}
\begin{aligned}
p^c_{10}(\delta,\epsilon)=&\min_{z} z\epsilon \sin2\theta +(1-\epsilon)(1-\delta)\\
&\text{ s. t. } z\geq \frac{1-\delta}{2+\sin2\theta}\,,
\end{aligned}
\end{equation*}
that is
\begin{equation}
p^c_{10}(\delta,\epsilon)=(1-\delta)\left[1-\frac{\epsilon}{1+\sin\theta\cos\theta}\right]\,.
\end{equation} This result could have been derived at an earlier stage than Lemma \ref{lemmafinal} (see Lemma \ref{lemmacomm} in the Appendix). In fact, it can also be derived without any reliance on the SDP formulation, by adapting the steps made in Ref.~\cite{qv_2qubit} to the case $\delta\neq 0$.

\subsection{Analytical solution for $\epsilon=1$}

Next, we present the analytical solution of \eqref{eq:final} for the special case $\epsilon=1$. The optimisation now reads
\begin{equation}\label{eq:finaleps1}
\begin{aligned}
p_{10}(\delta,1)=&\min_{z,x} \abs{x\cos2\theta+z\sin2\theta} \\
&\text{ s. t. } 1-\delta-z-\sqrt{x^2+z^2}\geq 0\\
&\phantom{\text{ s. t. }}1-\delta-z+\sqrt{x^2+z^2}\leq 1\\
&\phantom{\text{ s. t. }} 1-\delta-2z\leq\abs{x\cos2\theta+z\sin2\theta}\,.
\end{aligned}
\end{equation}
where we have spelled out the two matrix constraints in \eqref{eq:final}. Even for this simple case, the study is heavy, though without intrinsic difficulties.

First we notice that for the maximally entangled state ($\cos 2\theta=0$, $\sin2\theta=1$) the figure of merit is simply $z$, and the last constraint is $z\geq \frac{1-\delta}{3}$. The two quadratic constraints are both feasible for $z= \frac{1-\delta}{3}$, for a variety of values of $x$ including $x=0$. Thus, for $\theta=\frac{\pi}{4}$ we find $p_{10}(\delta,1)=p^c_{10}(\delta,1)=\frac{1-\delta}{3}$; both the commuting strategy and several non-commuting ones achieve this bound.

For $\cos 2\theta<1$, the solution is unique and can be inferred by studying the feasible region and the figure of merit graphically in the $(x,z)$ plane (Appendix \ref{appb}). The end result is:
\begin{equation}\label{eq:soleps1}
p_{10}(\delta,1)=p^c_{10}(\delta,1)\,+\,x^*\,\frac{2\cos 2\theta}{2+\sin 2\theta}
\end{equation} where $x^*=\max{(x_0,x_1)}$ is the optimal value of $x$ determined by
\begin{equation}\label{x0x1}
\begin{aligned} x_0&=(1-\delta)\,\left(\frac{\cos 2\theta-\sqrt{1+2\sin 2\theta}}{2+\sin 2\theta}\right)\,,\\
x_1&=-\frac{\delta\cos 2\theta+\sqrt{\delta^2(1+2\sin 2\theta)+2\delta(2+\sin 2\theta)}}{2+\sin 2\theta}\,.
\end{aligned}
\end{equation} Since both $x_0$ and $x_1$ are non-positive, $p_{10}(\delta,1)\leq p^c_{10}(\delta,1)$ as expected. Notice that \eqref{eq:soleps1} captures also the case $\theta=\frac{\pi}{4}$ (only, $x^*$ is not unique in that case). Besides, $x_0=0$ holds only for $\cos 2\theta=1$ i.e. for the product state; and $x_1=0$ holds only for $\delta=0$. In summary, for $\epsilon=1$, $0<\delta<1$, and $0<\theta<\frac{\pi}{4}$, the optimal strategy is \textit{not} the commuting one. 

\begin{figure}
  \includegraphics[width=\columnwidth]{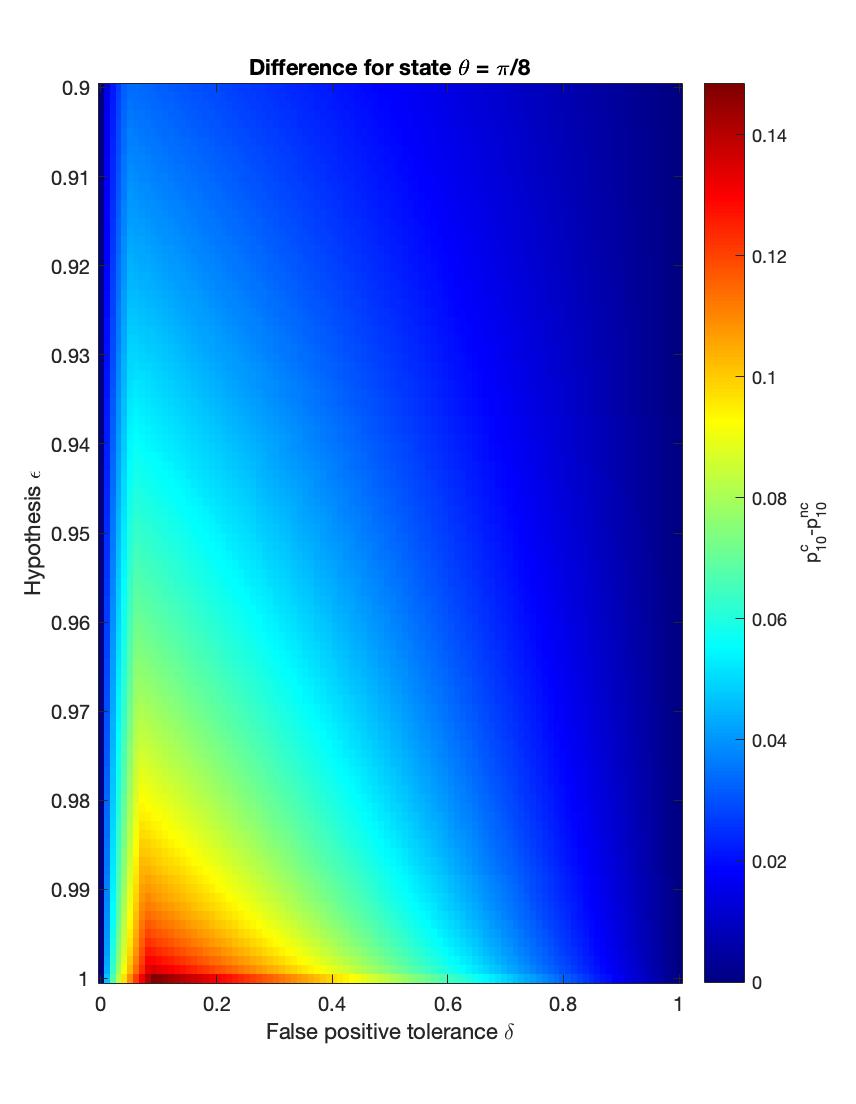}
  \includegraphics[width=\columnwidth]{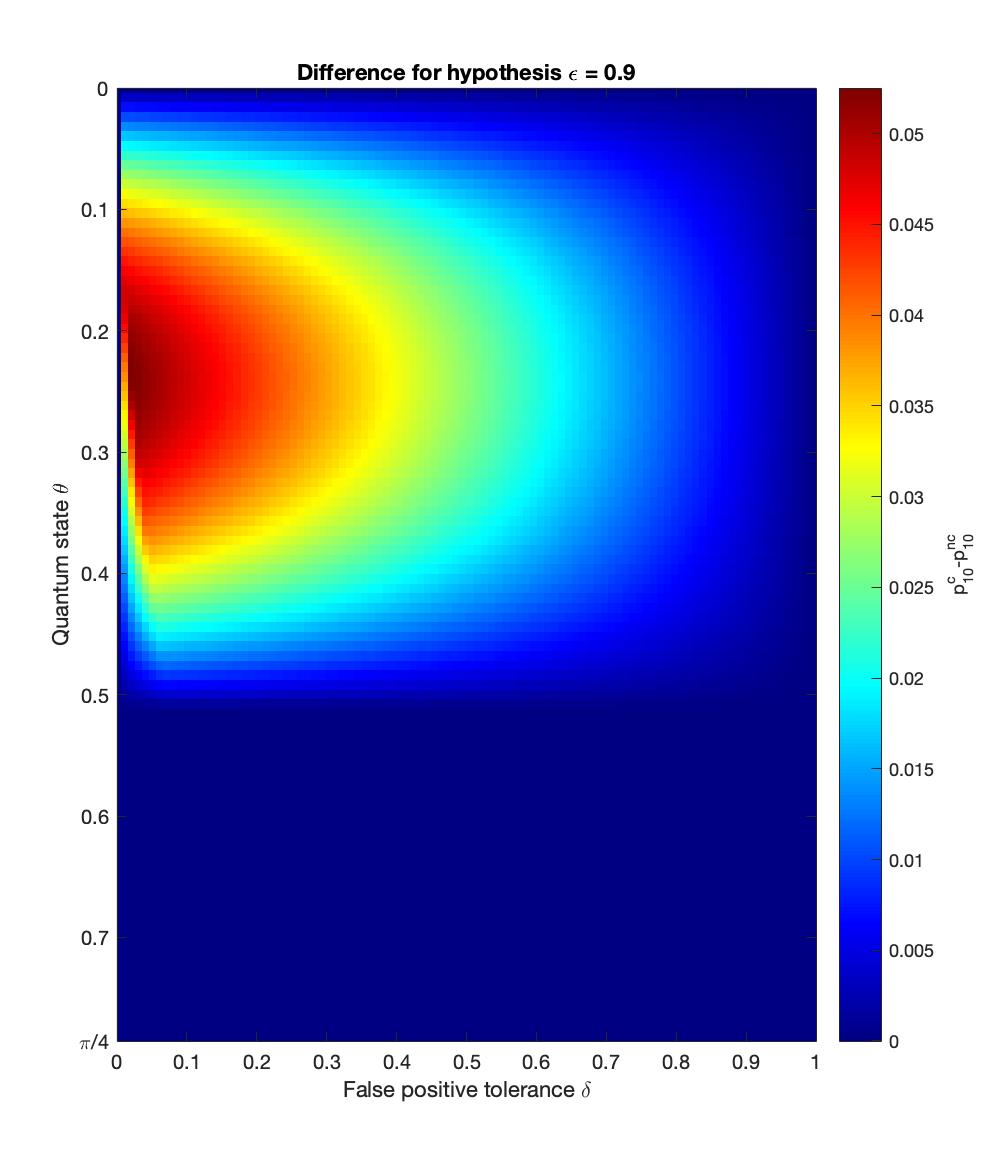}
  \caption{\label{fig:numerics}Difference $p^c_{10}(\delta,\epsilon)-p_{10}(\delta,\epsilon)$ between worst-case type II error probability for commuting and non-commuting strategies. Upper panel: in the plane $(\delta,\epsilon)$ for fixed state $\theta=\pi/8$. For values of $\epsilon\in[0,0.9]$ the difference is negligible and has been omitted from the plot. Lower panel: in the plane $(\delta,\theta)$ for fixed $\epsilon=0.9$.}
\end{figure}

Since we have set $\epsilon=1$, which means that $\sigma$ can be orthogonal to $\ket{\psi}$, it is natural to check also when $p_{10}(\delta,1)=0$, which would be obviously the case if we had not added the constraint that $\Omega$ must be separable. By inspection, we see that this is the case only for $\theta=0$, i.e. when the state itself is product (in which case it is trivial: one can find an orthogonal product state and check the orthogonality locally).

\subsection{Numerical solutions of the SDP}

We have seen that, even in the case $\epsilon=1$ when the figure of merit is at its simplest, the analytical solution requires some work and yields a not-so-transparent result. In view of this,  for arbitrary values of $\epsilon$ we leave aside any attempt of solving the optimisation \eqref{eq:final} analytically, and resort rather to numerical solutions of the SDP \eqref{eq:main_opt_sep_sdp}.

The results are presented in Fig.~\ref{fig:numerics}. We see that, in a large portion of parameter space, a commuting strategy is very close to being optimal (if not exactly so). A significant difference is seen only for $\epsilon\gtrsim 0.8$, that is, when the state $\sigma$ is allowed to be almost orthogonal to $\ket{\psi}$.

\section{\label{sec:conclusion}Conclusions}

We have worked in the quantum hypothesis testing scenario that has been called “quantum state verification”, in which the the null hypothesis is a pure state $\ket{\psi}$, while the alternative hypothesis may be any state $\sigma$ that is “distinguishable enough” from $\ket{\psi}$ (quantified by $\langle \psi|\sigma|\psi\rangle \leq 1-\epsilon$). Like previous works, we focused on entangled states shared by two distant players, and studied hypothesis testing under separable operations. We studied the tradeoff between the probability of false negative and that of false positive (the latter had been set to zero in previous studies, which amounts at assuming that the optimal POVM for the discrimination is implemented perfectly). The bilinear nature of the resulting optimization is overcome by reformulating the problem as a SDP. Then we presented the detailed solution for the case of two qubits, including analytical results for some extreme cases. We showed that, in general, the solution is a non-trivial modification of previous constructions: in particular, the optimal POVM may not commute with the closest state $\sigma$.

\acknowledgements{We thank Masahito Hayashi for discussions about our modifications of the hypotheses, and more generally the connection between quantum verification and hypothesis testing. This work is supported by the National Research Foundation and the Ministry of Education, Singapore, under the Research Centres of Excellence programme. L.P.T acknowledges support from the Alexander von Humboldt Foundation. M. D. acknowledges support from MEXT Quantum Leap Flagship Program (MEXT Q-LEAP) Grant No. JPMXS0118067285, JSPS KAKENHI Grant Number JP20K03774, and the International Research Unit of Quantum Information, Kyoto University.}

\bibliographystyle{plainnat}
\bibliography{bibliography}

\begin{appendix}
\section{Solution of the inner optimisation}
\label{sslengthy}

In this Appendix, we show how to go from Lemma \ref{lemma2} to Lemma \ref{lemmafinal}, while proving a few other intermediate results.

We begin by simplifying the inequality constraints on $t+z$ and $\omega$.
\begin{lem} Without loss of generality, the optimisation \eqref{eq:optlemma2} in Lemma  \ref{lemma2} becomes
\begin{equation}
\begin{aligned}
p_{10}(\delta,\epsilon)=&\min_{t,z,x,y_1,y_2} y_1+(1-\epsilon)y_2 \\
&\text{ s. t. } 0\preceq\begin{pmatrix}
t+z & x\\
x & t-z
\end{pmatrix}
\preceq\openone\\
&\phantom{\text{ s. t. }} t+z = 1-\delta\\
&\phantom{\text{ s. t. }} \begin{pmatrix}
y_1+y_2-(t+z) & -x\\
-x & y_1-(t-z)
\end{pmatrix}\succeq0 \\
&\phantom{\text{ s. t. }} y_1\geq\omega:=\abs{x\cos2\theta+z\sin2\theta} \\
&\phantom{\text{ s. t. }}  y_1\in\mathbb{R}, y_2\geq0 \\
\end{aligned}
\end{equation}
\end{lem}
\begin{proof}
With the notation introduced in the proof of Lemma \ref{lemma2}, for any feasible $(\bar{\Omega}_a,y_1,y_2)$ with $\bra{\psi}\bar{\Omega}_a\ket{\psi}>1-\delta\geq0$ for $\delta\in[0,1]$ there is another feasible $(\bar{\Omega}_a',y_1,y_2)$, where 
$$\bar{\Omega}_a' = \frac{(1-\delta)}{\bra{\psi}\bar{\Omega}_a\ket{\psi}}\bar{\Omega}_a$$
ensures $\bra{\psi}\bar{\Omega}_a'\ket{\psi}=1-\delta$, achieving the same objective value, we can without loss of generality assume that $t+z=1-\delta$.

It is clear that by reducing $\omega$, we increase the size of the feasible region of the inner optimisation over variables $y_1,y_2$. Therefore since $\abs{x\cos2\theta+z\sin2\theta}\leq1$ follows from other constraints, we have that $0\leq\abs{x\cos2\theta+z\sin2\theta}\leq\omega\leq1$, which means it suffices to take $\omega$ equal the lower bound.
\end{proof}

We now solve the inner optimisation, that is the optimisation over $y_1$ and $y_2$. It is natural to split our consideration into commuting strategy $x=0$ and non-commuting strategy $x\neq0$, as the commuting case is a simpler linear programming problem.
\begin{lem}\label{lem:inner_commuting}
For the commuting strategy $x=0$, the solution of the inner optimisation is
\begin{equation}
y_1^* + (1-\epsilon)\max\{0,t+z-y_1^*\}
\end{equation}
where $y_1^*:=\max\{t-z,\abs{z\sin2\theta}\}$.
\end{lem}
\begin{proof}
By Sylvester's criterion for psd, the inner minimization becomes
\begin{align*}
&\min_{y_1,y_2} y_1+(1-\epsilon)y_2 \\
&\text{ s. t. } y_1+y_2-(t+z)\geq0 \\
&\phantom{\text{ s. t. }} y_1-(t-z)\geq0 \\
&\phantom{\text{ s. t. }} (y_1+y_2-(t+z))(y_1-(t-z))\geq x^2 \\
&\phantom{\text{ s. t. }}  y_1\geq\abs{x\cos2\theta+z\sin2\theta}, y_2\geq0
\end{align*}
When $x=0$ the quadratic constraint trivially follows from the inequality constraints so we can drop it and the optimisation is linear. The constraints are
\begin{align*}
y_1&\geq\max\{t-z,\abs{x\cos2\theta+z\sin2\theta}\} \\
y_2&\geq\max\{0,(t+z)-y_1\}
\end{align*}
so that the minimum is reached at the lower bounds.
\end{proof}

\begin{lem}\label{lemmacomm}
For the commuting strategy $x=0$, the optimal error probability is given by
\begin{equation}\label{commapp}
p_{10}(\delta,\epsilon)=(1-\delta)\left[1-\frac{\epsilon}{1+\sin\theta\cos\theta}\right]
\end{equation}
\end{lem}
\begin{proof}
Since $x=0$, we are left with the program
\begin{align*}
&\min_{z} y_1^* + (1-\epsilon)\max\{0,t+z-y_1^*\} \\
&\text{ s. t. } 0\leq t-z\leq1\\
&\phantom{\text{ s. t. }} t+z = 1-\delta\\
&\phantom{\text{ s. t. }} y_1^*:=\max\{t-z,\abs{z\sin2\theta}\}
\end{align*}
The objective function can be rewritten as
$$\max\{y_1^*,(1-\epsilon)(1-\delta)+\epsilon y_1^*\}$$
from which we consider two cases. If $y_1^*\geq1-\delta$ then
\begin{align*}
&\min_{z} y_1^* \\
&\text{ s. t. } 0\leq 1-\delta-2z\leq1\\
&\phantom{\text{ s. t. }} y_1^*:=\max\{1-\delta-2z,\abs{z\sin2\theta}\}\geq1-\delta
\end{align*}
Here $y_1^*\geq0$ always, so the minimum is at least $1-\delta$.
If $y_1^*\leq1-\delta$ then
\begin{align*}
&\min_{z} (1-\epsilon)(1-\delta)+\epsilon y_1^* \\
&\text{ s. t. } 0\leq 1-\delta-2z\leq1\\
&\phantom{\text{ s. t. }} y_1^*:=\max\{1-\delta-2z,\abs{z\sin2\theta}\}\leq1-\delta
\end{align*}
It is straightforward to see that the minimum is achieved when
$$0\leq1-\delta-2z=\abs{z\sin2\theta}\leq1-\delta$$
corresponding to an optimal solution
$z^*=\frac{1-\delta}{2+\sin2\theta}$
with optimum value
$$(1-\delta)\left[1-\frac{\epsilon}{1+\sin\theta\cos\theta}\right]\,.$$
Since the global minimum is the smaller value of these two cases, the proof of the Lemma is complete.
\end{proof}

We remark that the structure of the optimal verification operator among all commuting strategies is rather simple. Explicitly we have that
\begin{align}
\Omega^* = \begin{pmatrix}
1-\delta & 0 & 0 & 0\\
0 & \omega^* & 0 & 0\\
0 & 0 & \omega^* & 0\\
0 & 0 & 0 & \omega^*
\end{pmatrix}\,,\quad \omega^*=\frac{(1-\delta)\sin2\theta}{2+\sin2\theta}\,.
\end{align}
This can be seen as a generalization of the optimal commuting strategy that Wang and Hayashi found for $\delta=0$ case~\cite{qv_2qubit} to the $\delta\in[0,1]$ case.

We now consider the noncommuting case, which is no longer a linear optimisation problem.
\begin{lem}\label{lem:inner_noncommuting}
For the noncommuting strategy $x\neq0$, the solution of the inner optimisation is
\begin{equation}
y_1^*+(1-\epsilon)\left[(t+z)-y_1^*+\frac{x^2}{y_1^*-(t-z)}\right]
\end{equation}
with the value
\begin{equation}\label{eq:y1opt_cases}
y_1^* = \begin{cases} \omega \,\mathrm{if}\, \hat{y}_1 \leq \omega \\
\hat{y}_1 \,\mathrm{if}\, \omega < \hat{y}_1 < t+\sqrt{x^2+z^2} \\
t+\sqrt{x^2+z^2} \,\mathrm{if}\, \hat{y}_1 > t+\sqrt{x^2+z^2}
\end{cases}
\end{equation}
for $\omega=\abs{x\cos2\theta+z\sin2\theta}$ and $\hat{y}_1=(t-z)+\sqrt{\frac{1-\epsilon}{\epsilon}}\abs{x}$.
\end{lem}
\begin{proof}
When $x\neq0$ (noncommuting strategy), the feasible region excludes the points $(y_1,y_2)$ where
\begin{align*}
y_1-(t-z) = 0, \text{ or } y_1+y_2-(t+z) = 0
\end{align*}
and so the optimisation becomes
\begin{align*}
&\min_{y_1,y_2} y_1+(1-\epsilon)y_2 \\
&\text{ s. t. } y_1>t-z, y_1\geq\omega \\
&\phantom{\text{ s. t. }} y_2\geq\max\left\{0,(t+z)-y_1+\frac{x^2}{y_1-(t-z)}\right\}
\end{align*}

Here the optimisation splits into two branches. Firstly, consider the branch
$$(t+z)-y_1+\frac{x^2}{y_1-(t-z)}\leq0$$
equivalently under the condition $y_1>t-z$
$$((t-z)-y_1)((t+z)-y_1)-x^2\geq0$$
and explicitly in terms of the roots
\begin{align*}
y_1&\leq \lambda_{\min}:=t-\sqrt{x^2+z^2} \text{ or },\\
 y_1&\geq \lambda_{\max}:=t+\sqrt{x^2+z^2}
\end{align*}
But then $t-\sqrt{x^2+z^2} < t-z < t+\sqrt{x^2+z^2}$ implies that the feasible region is $y_1\geq\max\{\omega,\lambda_{\max}\}$ leading to the optimum value $y_1^*=\max\{\omega,\lambda_{\max}\}$ which is always at least $\lambda_{\max}$. Secondly, the remaining branch
$$(t+z)-y_1+\frac{x^2}{y_1-(t-z)}\geq0\,,$$
which is equivalent to
\begin{align*}
t-\sqrt{x^2+z^2}=:\lambda_{\min} \leq y_1\leq \lambda_{\max}:=t+\sqrt{x^2+z^2}
\end{align*}
could be infeasible depending on $\omega$. However, whenever feasible, i.e. $\omega\leq\lambda_{\max}$, the minimum is upper bounded by the value of the objective function
$$y_1+(1-\epsilon)\left[(t+z)-y_1+\frac{x^2}{y_1-(t-z)}\right]$$
at the feasible point $y_1=\lambda_{\max}$, i.e. for which the objective value is $\lambda_{\max} + (1-\epsilon)*0$. Therefore, without loss of generality we consider this latter branch whenever feasible.

The inner optimisation becomes
\begin{align*}
&\min_{y_1} y_1+(1-\epsilon)\left[(t+z)-y_1+\frac{x^2}{y_1-(t-z)}\right]\\
&\text{ s. t. } y_1>t-z, \omega\leq y_1\leq\lambda_{\max}, \omega\leq\lambda_{\max}
\end{align*}
from which is is clear that the minimum is reached at the stationary point $\hat{y}_1$ which is the largest solution of 
$$\epsilon(\hat{y}_1-(t-z))^2=(1-\epsilon)x^2$$
whenever this point is feasible, or at the endpoints $\omega$ if $y_1^*<\omega$ and $\lambda_{\max}$ if $y_1^*>\lambda_{\max}$. (Note that $x\neq0$ ensures $\hat{y}_1>t-z$ if exists so that the smallest solution is always infeasible.)
\end{proof}


Finally, we present the proof of Lemma \ref{lemmafinal}:
\begin{proof}
With $y_1^*$ given before, we have to solve
\begin{equation}
\begin{aligned}
p_{10}(\delta,\epsilon)=&\min_{t,z,x} f(y_1^*) \\
&\text{ s. t. } 0\preceq\begin{pmatrix}
t+z & x\\
x & t-z
\end{pmatrix}
\preceq\openone\\
&\phantom{\text{ s. t. }} t+z = 1-\delta
\end{aligned}
\end{equation}
This becomes an optimisation over two real variables $z,x$ after eliminating $t$. To see the branch reduction, we consider feasible $(z,x)$ that satisfies
\begin{equation}\label{eq:feasible_zx}
0\preceq\begin{pmatrix}
1-\delta & x\\
x & 1-\delta-2z
\end{pmatrix}
\preceq\openone
\end{equation}
and show that objective value (abuse of notation and redefine the function $f$ eliminating variable $t$)
\begin{equation*}
f(y_1^*):=y_1^*+(1-\epsilon)\left[1-\delta-y_1^*+\frac{x^2}{y_1^*-(1-\delta-2z)}\right]
\end{equation*}
is smaller in the region {\bf I} defined by
$$1-\delta-2z+\sqrt{\frac{1-\epsilon}{\epsilon}}\abs{x}\leq\abs{x\cos2\theta+z\sin2\theta}\,.$$

In the region {\bf III} defined by the inequality
$$1-\delta-2z+\sqrt{\frac{1-\epsilon}{\epsilon}}\abs{x}\geq1-\delta-z+\sqrt{x^2+z^2}$$
the objective function takes the value
\begin{align*}
&f(1-\delta-z+\sqrt{x^2+z^2})\\
&=(1-\delta)+\epsilon(-z+\sqrt{x^2+z^2})+\frac{(1-\epsilon)x^2}{z+\sqrt{x^2+z^2}}\\
&=(1-\delta)+\epsilon(-z+\sqrt{x^2+z^2})-(1-\epsilon)(z-\sqrt{x^2+z^2})\\
&=1-\delta-z+\sqrt{x^2+z^2}
\end{align*}
which is a function of two independent variables $z,x$, and is increasing in terms of $\abs{x}$ for a fixed value of $z$.

Likewise, in the region {\bf II} defined by the inequality
\begin{align*}
\abs{x\cos2\theta+z\sin2\theta} &\leq 1-\delta-2z+\sqrt{\frac{1-\epsilon}{\epsilon}}\abs{x}\\
&\leq 1-\delta-z+\sqrt{x^2+z^2}
\end{align*}
the objective function take the value
\begin{align*}
&f\left(1-\delta-2z+\sqrt{\frac{1-\epsilon}{\epsilon}}\abs{x}\right)\\
&=(1-\epsilon)(1-\delta)+\epsilon\left(1-\delta-2z+\sqrt{\frac{1-\epsilon}{\epsilon}}\abs{x}\right)+\epsilon\\
&=(1-\delta)+\epsilon-2\epsilon z+\sqrt{\epsilon(1-\epsilon)}\abs{x}\,,
\end{align*}
which is also increasing in $\abs{x}$. Moreover, at the boundary between two regions, the objective functions agree.

The argument now goes as follows: for each feasible $z$, we look at the set of feasible $x$ that is defined by~\eqref{eq:feasible_zx}. For any feasible $x_1,x_2$ in region {\bf III} (if exist), since the objective function is increasing, the point with smaller $\abs{x_j}$ achieves a lower objective value. Hence for minimization, it suffices to consider feasible $x$ in the boundary of region {\bf III}. Since this boundary is also contained in region {\bf II}, we have shown that without loss of generality it suffices to consider the feasible $x$ that belong to regions {\bf I} and {\bf II}. Now the argument can be repeated: points $x_3,x_4$ in region {\bf II} with smaller $\abs{x_j}$ achieve small objective value. This reduces the feasible region to region {\bf I} only.
\end{proof}

\section{The optimal solution for $\epsilon=1$}
\label{appb}

In this Appendix, we proceed to solve \eqref{eq:finaleps1}.

\begin{figure}[h]
  \includegraphics[width=\columnwidth]{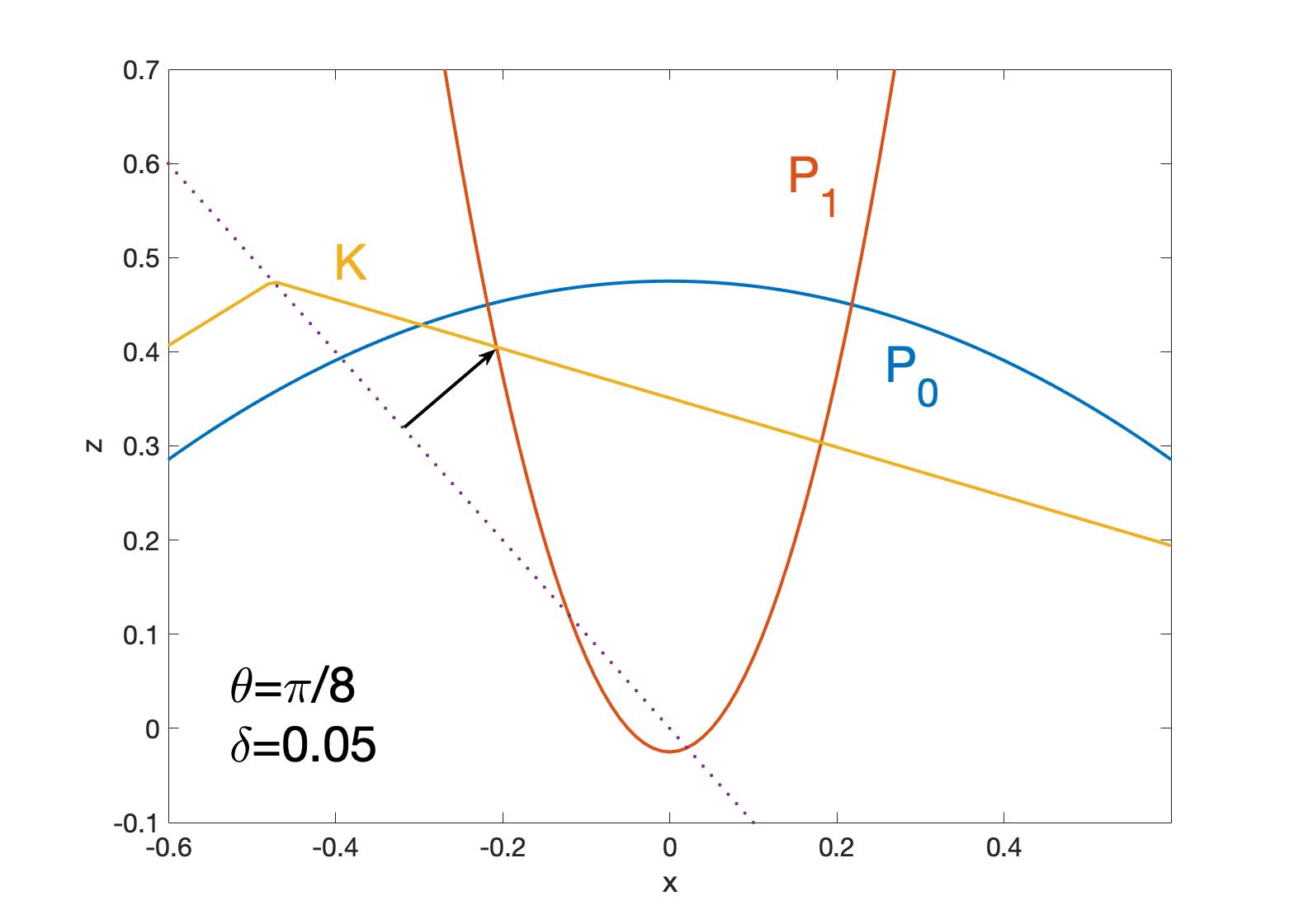}
  \includegraphics[width=\columnwidth]{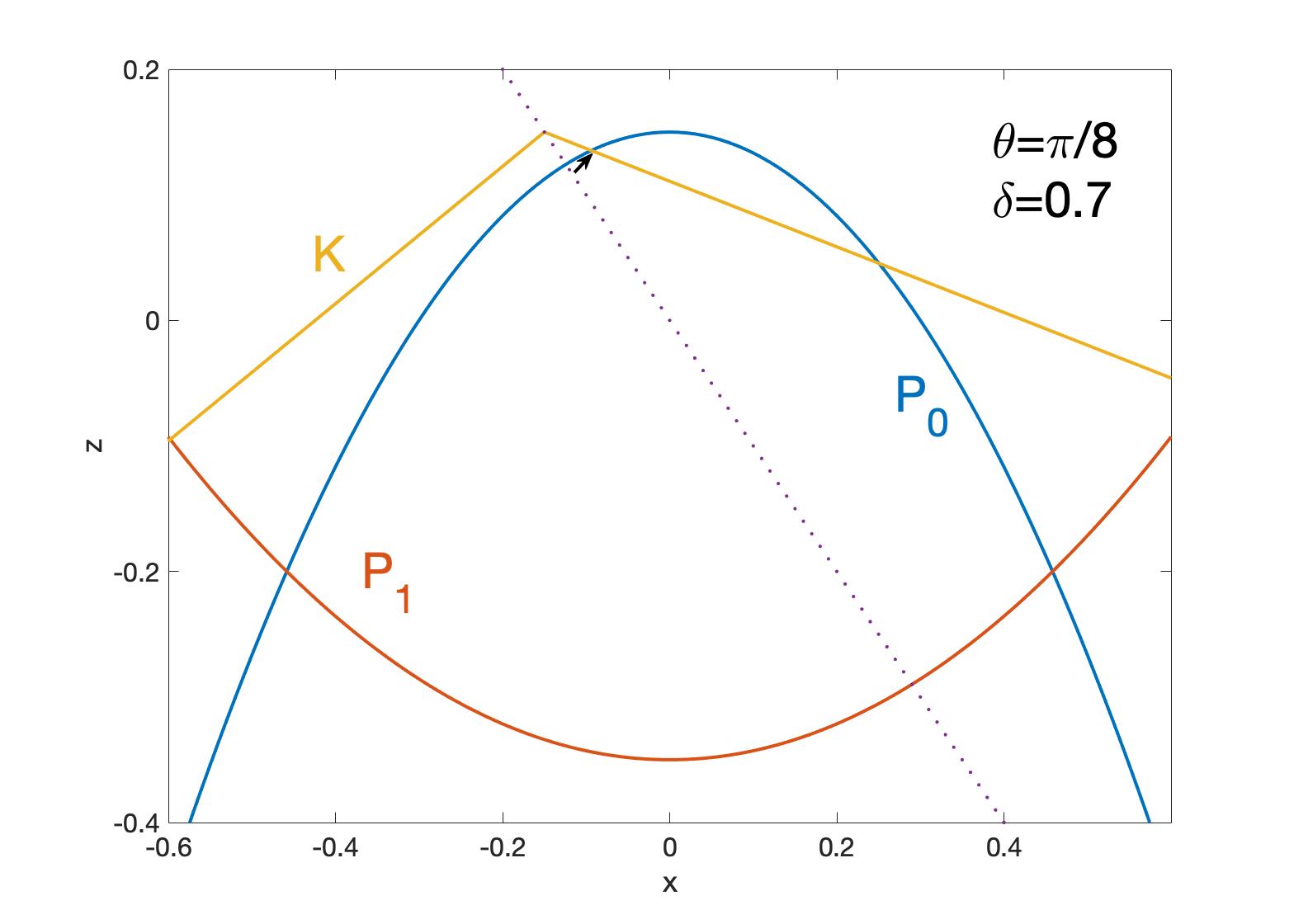}
  \caption{Graphic determination of the solution of the optimisation \eqref{eq:finaleps1}. These two examples are plotted for $\theta=\frac{\pi}{8}$ and two values of $\delta$. The solid lines are the boundaries of the regions defined in Eqs \eqref{eq:P0}, \eqref{eq:P1} and \eqref{eq:K}. The dotted line is $z=-\cot{2\theta}x$: the solution of the optimisation is the point of the feasible region that is closest to this line, in either direction. For small $\delta$, the solution is at the intersection $x_1<0$ of $\mathbf{P_1}$ and $\mathbf{K}$; for large $\delta$, at the intersection $x_0<0$ of $\mathbf{P_0}$ and $\mathbf{K}$.}\label{fig4}
\end{figure}

The feasible region is the intersection of three regions in the $(x,z)$ plane:
\begin{itemize}
\item The constraint $1-\delta-z-\sqrt{x^2+z^2}\geq 0$ defines the region \begin{align}\label{eq:P0}\mathbf{P_0}:&\;z\,\leq\,-\frac{x^2}{2(1-\delta)}+\frac{1-\delta}{2}\,,\end{align} upper-bounded by a parabola whose maximum at $(x,z)=(0,\frac{1-\delta}{2})$.
\item The constraint $1-\delta-z+\sqrt{x^2+z^2}\leq 1$ defines the region \begin{align}\label{eq:P1}\mathbf{P_1}:&\;z\,\geq\, \frac{x^2}{2\delta}-\frac{\delta}{2}\,,\end{align} lower-bounded by parabola whose minimum is at $(x,z)=(0,-\frac{\delta}{2})$
\item The constraint $1-\delta-2z\leq\abs{x\cos2\theta+z\sin2\theta}$ defines the region
\begin{align}\label{eq:K}
\mathbf{K}:&\;\left\{\begin{array}{lcl}
z\,\geq\, \frac{1-\delta+x\cos 2\theta}{2-\sin 2\theta}&\textrm{ for } &x\leq x_k  \\
z\,\geq\, \frac{1-\delta-x\cos 2\theta}{2+\sin 2\theta}&\textrm{ for }  &x\geq x_k
\end{array}\right.
\end{align} lower-bounded by a broken line with kink at $x_k=-\frac{1-\delta}{2}\tan(2\theta)$ the intersection with $x\cos2\theta+z\sin2\theta=0$.
\end{itemize}
The figure of merit to be minimised is $|x\cos2\theta+z\sin2\theta|$: this means that $p_{10}(\delta,\epsilon)$ is given by the smallest distance between the line $z=-\cot{2\theta}x$ and a point of the feasible region. A graphical inspection (see Figs \ref{fig4} and \ref{fig5}) shows that this minimal distance is always achieved by the point that is the intersection of either $\mathbf{P_0}$ or $\mathbf{P_1}$ with the line $z\,\geq\, \frac{1-\delta-x\cos 2\theta}{2+\sin 2\theta}$. These are the points whose $x$ coordinates have been called $x_0$ and $x_1$, given in Eq.~\eqref{x0x1} in the main text. A more fully analytical proof of this result would not bring further clarity (and for good measure, the correctness of the result has been double-checked numerically with the solution of the corresponding SDP).

\begin{figure}[h]
  \includegraphics[width=\columnwidth]{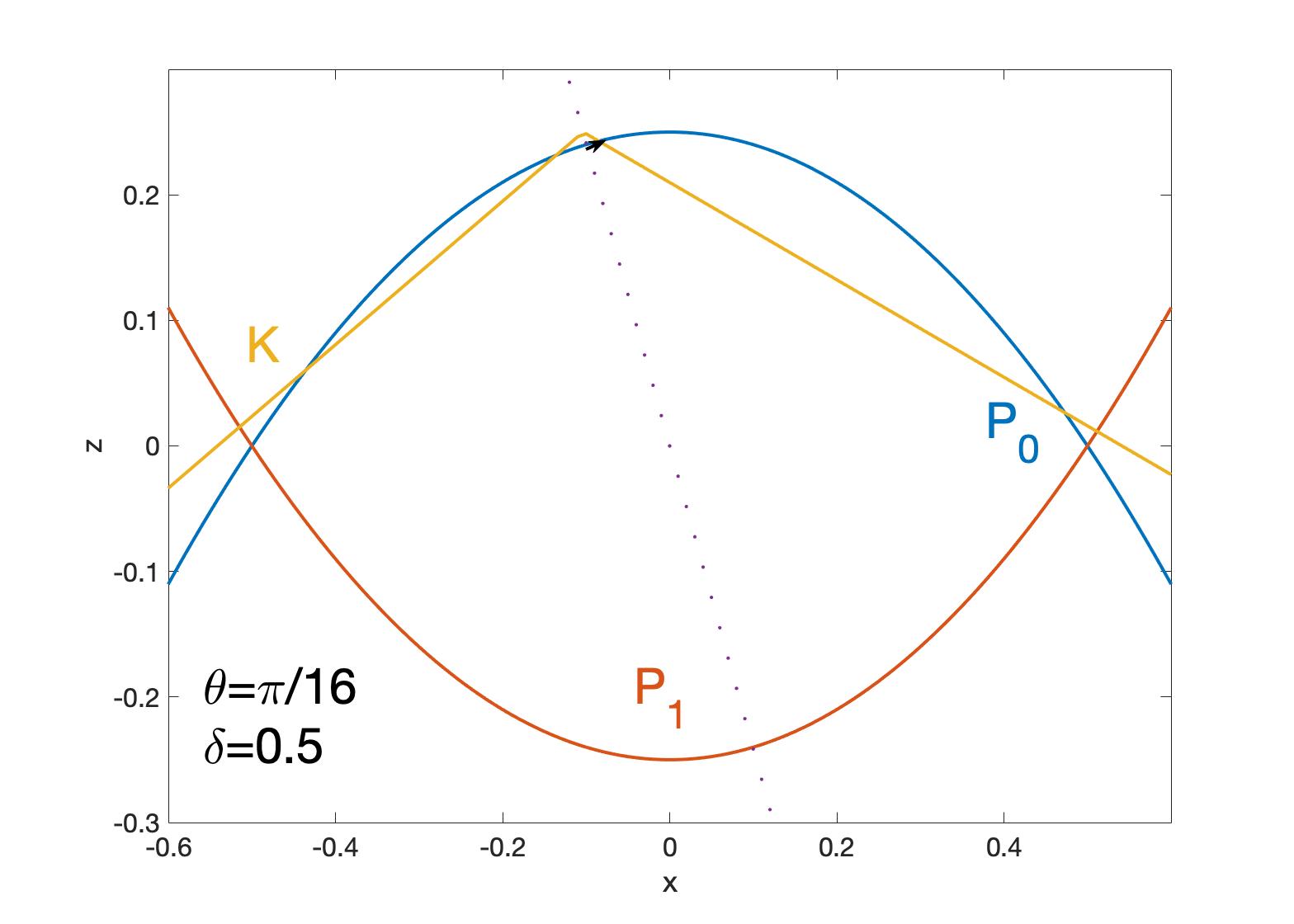}
  \caption{When $\theta$ is reduced, the slope of the left segment of $K$ increases and also cuts $P_0$, whence the feasible region consists of two disjoint sets. Nonetheless, the closest point to the line $z=-\cot{2\theta}x$ remains the one on the right segment of $K$.}\label{fig5}
\end{figure}

\end{appendix}

\end{document}